\documentclass[epj,nopacs]{svjour}

\usepackage{graphicx}
\usepackage{amsmath,amssymb,bm}
\usepackage{booktabs,tabularx,array}
\usepackage{microtype}
\usepackage{placeins}
\usepackage{url}
\usepackage{cite}
\newcommand{\Mpl}{M_{\rm Pl}}

\newcolumntype{Y}{>{\raggedright\arraybackslash}X}

\begin{document}

\title{Action-level characterization of gravitational-wave propagation in dynamical Barbero--Immirzi gravity}
\titlerunning{Action-level GW propagation in dynamical BI gravity}

\author{Zhi-Fu Gao\inst{1} \and
Hui Wang\inst{2} \and
Luiz Carlos Garcia de Andrade\inst{3} \and
Na Wang\inst{1} \and
Guo-Qiang Jin\inst{4} \and
Zhou-Jian Cao\inst{5,6}}
\authorrunning{Z.-F. Gao et al.}

\offprints{Hui Wang}

\institute{State Key Laboratory of Radio Astronomy and Technology, Xinjiang Astronomical Observatory, Chinese Academy of Sciences, Urumqi 830011, China
\and Shanxi Province Intelligent Optoelectronic Sensing Application Technology Innovation Center, Yuncheng University, Yuncheng, China, \email{wanghuiycu@163.com}
\and Departamento de F\'isica Te\'orica, IF-UERJ, Rio de Janeiro, Brazil
\and College of Mechanical and Electrical Engineering, Tarim University, Alar 843300, Xinjiang, China
\and Institute for Frontiers in Astronomy and Astrophysics, Beijing Normal University, Beijing 102206, China
\and School of Physics and Astronomy, Beijing Normal University, Beijing 100875, China}

\date{Received: date / Revised version: date}

\abstract{We investigate which operators in first-order dynamical Barbero--Immirzi (BI) gravity control cosmological tensor propagation at the action level. Within a bosonic, two-derivative, curvature-linear Einstein--Cartan class, eliminating the algebraic Lorentz connection reveals a separation between the scalar sector and the tensor kinetic normalization. In particular, the Holst-to-Palatini ratio determines the scalar kinetic structure, whereas the transverse-traceless tensor mode is normalized by the parity-even Hilbert--Palatini coefficient. As a consequence, a minimal dynamical Holst sector with fixed parity-even normalization remains exactly on the general-relativistic tensor-propagation surface and does not generate anomalous gravitational-wave friction. We then construct a regular nonminimal parity-even realization in which the tensor normalization evolves cosmologically, derive the corresponding standard-siren observable, and apply current GWTC-3, GWTC-4.0, and GWTC-5.0 information as an observational test of this action-level tensor sector. The resulting constraints should therefore be interpreted as limits on the evolution of the tensor kinetic normalization rather than direct constraints on minimal BI torsion dynamics.}

\maketitle

\section{Introduction}

A central question in first-order modified gravity is not merely whether extra fields are present, but which operator in the fundamental action controls the propagating tensor mode. In dynamical Barbero--Immirzi gravity, the presence of a scalar BI degree of freedom does not automatically imply a modified gravitational-wave propagation law. The relevant observable is determined by the tensor kinetic normalization emerging after the auxiliary Lorentz connection is eliminated.

\label{sec:intro}

Einstein--Cartan gravity treats the tetrad and Lorentz connection as
independent variables.  In the minimal theory torsion is nondynamical and is
fixed algebraically by spin \cite{Kibble1961,Hehl1976}.  The Holst term adds a
parity-odd contraction of the curvature.  Its constant coefficient, usually
written in terms of the Barbero--Immirzi (BI) parameter, does not change the
vacuum Einstein equations \cite{Holst1996,Barbero1995,Immirzi1997}, although it
can enter fermionic contact interactions \cite{Freidel2005,Perez2006}.

Promoting the inverse BI parameter to a pseudoscalar makes its gradient a
source of torsion.  Because the connection remains algebraic, eliminating it
gives general relativity (GR) plus a pseudoscalar with a field-dependent
kinetic term \cite{Taveras2008,TorresGomez2009}.  The related Nieh--Yan
construction gives a canonical pseudoscalar and clarifies its matter
interactions \cite{Calcagni2009,MercuriTaveras2009,Mercuri2009}.  This
distinction is decisive for gravitational-wave (GW) propagation: torsion can
alter the homogeneous stress tensor without changing the coefficient of the
tensor kinetic operator.  A background scalar therefore does not by itself
imply an anomalous GW friction term.

General first-order scalar theories make this operator distinction explicit.
Einstein--Cartan actions with field-dependent Palatini, Holst, and Nieh--Yan
coefficients have been reduced systematically to metric effective theories
\cite{Shaposhnikov2020,Karananas2021,Langvik2021,Rigouzzo2022}; derivative
expansions and quantum treatments likewise identify the independent torsion
operators \cite{Diakonov2011,Shapiro2014}.  In metric language, a running
tensor normalization belongs to the scalar--tensor sector
\cite{Bellini2014}, whose luminal subset became especially relevant after the
binary-neutron-star speed bound \cite{Ezquiaga2017}.

Standard sirens constrain a running tensor normalization through the ratio
of GW and electromagnetic luminosity distances
\cite{Nishizawa2018,Belgacem2018a,Belgacem2018b}.  Forecasts and catalog
analyses have developed this observable across detector bands and source
populations \cite{Belgacem2019,Lagos2019,AbbottGW170817,Finke2021,
Mastrogiovanni2021,Mukherjee2021}.  The commonly used phenomenological parameterization introduces a ratio between
gravitational-wave and electromagnetic luminosity distances, together with an
asymptotic high-redshift value $\Xi_0$ and a transition index $n$; the
GR limit is $\Xi_0=1$.  GWTC-3 analyses established broad and prior-sensitive
bounds \cite{Mancarella2022,Leyde2022,ChenGray2024,GWTC3}.  The GWTC-4.0
catalog and cosmology analysis substantially enlarged the data set and released
pipeline-level hyperposterior samples
\cite{GWTC4catalog,GWTC4cosmology,GWTC4data}.  The subsequent GWTC-5.0
spectral-siren analysis used 235 GW candidates for its modified-propagation
tests and reported tighter $\Xi_0$ constraints, with no evidence for a
departure from GR \cite{GWTC5cosmology}.

Complementary studies test modified friction with precessing binaries and
parity- or Lorentz-violating propagation \cite{Lin2025,Zhu2024}.  Earlier
Immirzi-field work also studied vacuum GW polarizations and test-particle
response in torsionful first-order or Palatini models
\cite{Bombacigno2018,Bombacigno2019}; those observables are distinct from the
cosmological amplitude damping of the metric tensor modes considered here.
BI and Einstein--Cartan phenomenology further includes baryogenesis, waveform
generation, torsion masses, and fermion condensates
\cite{Aliberti2023,Battista2021,GarciaGao2026,GarciaFermion2026}.  Recent
JHEP analyses of Einstein--Cartan gravity emphasize the importance of deriving
the propagating field content only after the connection has been eliminated
\cite{Karananas2024}.  Future population and large-scale-structure methods
may reduce current degeneracies \cite{Bertheas2026,DeLeo2026,Nanadoumgar2026}.

An earlier version of the present work employed a phenomenological
propagation-equation-level friction ansatz for the minimal BI sector, while a
related BI--quintessence extension was explored separately
\cite{GaoQuintessence2026}.  The action reduction developed here supersedes
that minimal-sector interpretation: an effective friction cannot be attributed
to the minimal two-derivative dynamical-Holst action alone.  Within the action
class studied here it must instead be matched to an evolving parity-even
curvature coefficient.  This distinction is essential when translating a
standard-siren bound into a statement about microscopic BI dynamics.

The central question is therefore an operator question rather than a
phenomenological one: which coefficient in a first-order BI action actually
normalizes the propagating tensor mode?  This work answers that question at
the action level.  First, we derive the tensor action of the minimal
dynamical-Holst theory and establish an exact propagation result: a dynamical Holst
coefficient at fixed Hilbert--Palatini normalization does not generate
anomalous GW friction.  Second, for the curvature-linear first-order Einstein--Cartan
class in Eq.~\eqref{eq:mother_action} we show that the parity-even Hilbert--Palatini coefficient alone
sets the tensor normalization, while the Holst-to-Palatini ratio fixes a
specific torsion contribution to the scalar kinetic metric.  The minimal-sector result is thus a corollary of a broader propagation theorem, rather than
a property of one background solution.  We then introduce the leading
parity-invariant nonminimal completion, construct a regular cosmological
realization satisfying the invariant two-derivative kinetic conditions by
integrating the Raychaudhuri equation and
reconstructing its potential, and use standard-siren results only as an
observational application of the action-level map.  For this purpose we
map the published GWTC-3, GWTC-4.0, and GWTC-5.0 intervals to the corresponding
tensor-normalization variables.  For GWTC-4.0, for which the two named public
hyperposterior products used in our reproducibility archive are available, we
also transform the samples one by one, retain their parameter covariances and
prior measure, and keep the final catalog result separate from the two
versioned public pipeline products.

The remainder of this paper is organized as follows.  Section~\ref{sec:minimal}
derives the exact propagation result for the minimal dynamical-Holst
sector.  Section~\ref{sec:completion} develops the first-order BI curvature
completion and its tensor dynamics, while Sec.~\ref{sec:background} tests
background consistency and constructs a regular benchmark.  The observational
analysis is presented in Secs.~\ref{sec:data} and
\ref{sec:posterior_feasibility}, where the action-level prediction is connected
to current standard-siren constraints, published GWTC-3 through GWTC-5.0
results are recast, the public GWTC-4.0 pipelines are compared, and the
remaining propagation--population
degeneracies are examined.  Section
\ref{sec:discussion} places the results in the context of earlier BI and
modified-propagation studies.  Section~\ref{sec:conclusion} summarizes the
conclusions and outlines the observational and theoretical steps needed for a
model-specific test.  Technical derivations and reproducibility details are
collected in the appendices.

We use signature $(-,+,+,+)$, $\epsilon_{0123}=+1$, and the reduced Planck
mass $\Mpl$.  A prime on a background quantity denotes $d/dN$ with
$N\equiv\ln a$, except where a derivative is explicitly labelled by its
argument.

\section{Minimal dynamical Holst gravity: exact GR-like tensor propagation}
\label{sec:minimal}

Let $e^I$ be a tetrad and $\omega^{IJ}$ an independent Lorentz connection,
with curvature $R^{IJ}(\omega)$.  A convenient convention for the minimal
dynamical inverse-BI field $\beta$ is
\begin{equation}
\begin{aligned}
 S_{\rm min}={}&\frac{\Mpl^2}{4}\int
 \left(\epsilon_{IJKL}+2\beta\,\eta_{I[K}\eta_{L]J}\right)
 e^I\wedge e^J\wedge R^{KL}(\omega)\nonumber\\
 &-\int d^4x\sqrt{-g}\,V(\beta)+S_m[g,\Psi].
\end{aligned}
 \label{eq:minimal_first_order}
\end{equation}
The choice of $\beta$ or its reciprocal is conventional and does not affect the
argument below.  Writing $\omega=\mathring\omega(e)+C$, the connection
equation is algebraic.  In the absence of fermionic spin sources its solution
is
\begin{equation}
 C_{\mu IJ}=-\frac{1}{2(1+\beta^2)}
 \left(\epsilon_{IJ}{}^{KL}e_{\mu K}\partial_L\beta
 -2\beta e_{\mu[I}\partial_{J]}\beta\right).
 \label{eq:minimal_contorsion}
\end{equation}
Substitution into Eq.~\eqref{eq:minimal_first_order}, including the quadratic
contorsion terms, gives \cite{Taveras2008,TorresGomez2009}
\begin{equation}
\begin{aligned}
 S_{\rm min}^{\rm eff}={}&\int d^4x\sqrt{-g}\left[
 \frac{\Mpl^2}{2}R
 -\frac{3\Mpl^2}{4(1+\beta^2)}(\partial\beta)^2
 -V(\beta)\right]\nonumber\\
 &+S_m[g,\Psi].
\end{aligned}
 \label{eq:minimal_metric}
\end{equation}
The field redefinition
$d\varphi/d\beta=\sqrt{3/2}\,\Mpl/(1+\beta^2)^{1/2}$ makes the scalar kinetic
term canonical.  Equation~\eqref{eq:minimal_metric} is therefore GR with a
minimally coupled pseudoscalar, not a theory with a running Planck mass.

For a flat FLRW background and transverse-traceless perturbations
$g_{ij}=a^2(\delta_{ij}+h_{ij})$, the quadratic tensor action is
\begin{equation}
\begin{aligned}
 S_T^{(2)}={}&\frac{1}{8}\sum_{\lambda=+,\times}\int dt\,d^3x\,
 a^3 Q_T\left[\dot h_\lambda^2-
 \frac{c_T^2}{a^2}(\boldsymbol\nabla h_\lambda)^2\right],\nonumber\\
 &Q_T=\Mpl^2,\qquad c_T^2=1.
\end{aligned}
 \label{eq:minimal_tensor_action}
\end{equation}
Each Fourier mode obeys
\begin{equation}
 \ddot h_\lambda+3H\dot h_\lambda+\frac{k^2}{a^2}h_\lambda=0.
 \label{eq:minimal_wave}
\end{equation}
Consequently,
\begin{equation}
 \nu\equiv\frac{d\ln Q_T}{d\ln a}=0,
 \qquad \frac{d_L^{\rm gw}}{d_L^{\rm em}}=1.
 \label{eq:no_go}
\end{equation}
This conclusion is nonperturbative with respect to the homogeneous field value and velocity of
$\beta$ within the two-derivative action \eqref{eq:minimal_first_order}.  The
BI stress tensor can change $H(z)$ and hence the common electromagnetic and GW
distance, but it cannot create a relative amplitude damping.  This is the structural propagation result used throughout the remainder of the paper.

The result already indicates the structural origin of this propagation property: the
Holst coefficient can alter the torsion solution and the scalar sector without
changing the parity-even coefficient multiplying the tensor kinetic operator.
Section~\ref{sec:completion} now makes this statement precise for the
curvature-linear first-order action class considered here.

\section{Curvature-linear first-order completion and the propagation criterion}
\label{sec:completion}

\subsection{Mother action and elimination of the connection}

To obtain a running tensor normalization, the coefficient of the
Hilbert--Palatini term must itself vary.  We consider the parity-invariant
first-order action
\begin{align}
 S_{\rm 1st}={}&\frac{\Mpl^2}{4}\int
 \left[A(x)\epsilon_{IJKL}+2B(x)\eta_{I[K}\eta_{L]J}\right]\nonumber\\
 &\hspace{1.2em}\times e^I\wedge e^J\wedge R^{KL}(\omega) \nonumber\\
 &-\int d^4x\sqrt{-g}\left[\frac{\Mpl^2}{2}Z(x)(\partial x)^2+U(x)\right]\nonumber\\
 &+S_m[g,\Psi].
 \label{eq:mother_action}
\end{align}
Here $x$ is a dimensionless BI pseudoscalar; parity requires $A$, $Z$, and
$U$ to be even and $B$ to be odd.  The ratio
\begin{equation}
 \gamma(x)\equiv \frac{B(x)}{A(x)}
 \label{eq:gamma_def}
\end{equation}
is the field-dependent Holst-to-Palatini coefficient.  Actions of this form
are a restricted, bosonic sector of general Einstein--Cartan scalar theories
\cite{Shaposhnikov2020,Karananas2021,Langvik2021,Rigouzzo2022}.

After the Weyl transformation $g^E_{\mu\nu}=A(x)g_{\mu\nu}$, the connection
equation has the same algebraic form as Eq.~\eqref{eq:minimal_contorsion} with
$\beta\rightarrow\gamma(x)$.  Eliminating the connection gives
\begin{align}
 S_E={}&\int d^4x\sqrt{-g_E}\bigg[\frac{\Mpl^2}{2}R_E\nonumber\\
 &\hspace{2.0em}-\frac{\Mpl^2}{2}k_E(x)(\partial_E x)^2
 -\frac{U(x)}{A(x)^2}\bigg]\nonumber\\
 &+S_m[A^{-1}g_E,\Psi],
 \label{eq:Einstein_action}\\
 k_E(x)={}&\frac{Z(x)}{A(x)}+
 \frac{3}{2}\frac{\gamma_{,x}^2}{1+\gamma(x)^2}.
 \label{eq:kE_general}
\end{align}
The second term in Eq.~\eqref{eq:kE_general} is the calculable remnant of the
algebraic torsion.  Thus the same first-order operator that identifies $x$ as
a BI field fixes part of its effective kinetic metric; the construction is
not obtained by merely relabelling an arbitrary metric scalar.

Transforming Eq.~\eqref{eq:Einstein_action} back to the matter (Jordan) frame
gives
\begin{equation}
\begin{aligned}
 S_J={}&\int d^4x\sqrt{-g}\left[\frac{\Mpl^2}{2}A(x)R
 -\frac{\Mpl^2}{2}k_J(x)(\partial x)^2-U(x)\right]\nonumber\\
 &+S_m[g,\Psi],
\end{aligned}
 \label{eq:Jordan_action}
\end{equation}
with
\begin{equation}
 k_J=A\left[k_E-\frac{3}{2}\left(\frac{A_{,x}}{A}\right)^2\right].
 \label{eq:kJ}
\end{equation}
For $A>0$, the Einstein-frame action makes the invariant two-derivative
kinetic criterion transparent: the tensor kinetic coefficient is positive and
the scalar is non-ghost if $k_E>0$.  The sign of $k_J$ by itself is not a
frame-invariant no-ghost criterion because the Jordan-frame operator $A(x)R$
contains scalar--metric kinetic mixing.  Equivalently,
$k_E=k_J/A+\tfrac32(A_{,x}/A)^2$, so a negative $k_J$ need not imply a
negative physical scalar kinetic eigenvalue.  We will nevertheless use $k_J>0$
later as an \emph{additional} restriction defining one diagnostic subclass;
the explicit benchmark happens to satisfy it.

\subsection{Canonical coupling and an action-level propagation criterion}

The elimination above separates two roles that are sometimes conflated in
phenomenological treatments.  The ratio $\gamma=B/A$ determines the algebraic
torsion and therefore contributes to the scalar kinetic metric in
Eq.~\eqref{eq:kE_general}.  By contrast, the parity-even coefficient $A(x)$
multiplies the metric curvature operator and hence controls the normalization
of the propagating tensor mode.  The formal theorem and its proof are given in
Appendix~\ref{app:propagation_theorem}; here we quote the result in the form
used throughout the main text.

For the action~\eqref{eq:mother_action}, assuming a spatially flat FLRW
background, a homogeneous $x(t)$, bosonic matter minimally coupled to the
metric, and regular algebraic elimination of the Lorentz connection, the exact
two-derivative quadratic action for linear transverse-traceless perturbations
is
\begin{equation}
 S_T^{(2)}=\frac{\Mpl^2}{8}\sum_\lambda\int dt\,d^3x\,
 a^3A(x)\left[\dot h_\lambda^2-
 \frac{(\boldsymbol\nabla h_\lambda)^2}{a^2}\right].
 \label{eq:prop_theorem_action}
\end{equation}
Therefore the tensor kinetic normalization, tensor speed, and friction
parameter are
\begin{equation}
 Q_T=\Mpl^2 A(x),\qquad c_T^2=1,\qquad
 \nu=\frac{d\ln A}{d\ln a}.
 \label{eq:prop_theorem_local}
\end{equation}
Here $Q_T$ is the coefficient of the tensor kinetic term,
$c_T$ is the propagation speed of the tensor mode, and
$\nu$ characterizes the extra Hubble-friction contribution in the wave
equation.  In the geometric-optics limit, if gravitational-wave and
electromagnetic luminosity distances are defined in the same Jordan-frame
background and we write $A(z)\equiv A[x(z)]$ and $A_0\equiv A[x(z=0)]$, then
\begin{equation}
 \frac{d_L^{\rm gw}(z)}{d_L^{\rm em}(z)}=
 \sqrt{\frac{A_0}{A(z)}}.
 \label{eq:prop_theorem_distance}
\end{equation}
This relation makes the physical separation transparent: $B(x)$ can influence
GW observables indirectly through the background solution $x(z)$, but within
Eq.~\eqref{eq:mother_action} it does not introduce an independent tensor
kinetic, gradient, or friction operator.

The minimal bosonic two-derivative dynamical-Holst theory is the special case
$A=1$.  It therefore satisfies $Q_T=\Mpl^2$, $c_T=1$, and $\nu=0$, so the
propagation-defined luminosity distances obey $d_L^{\rm gw}=d_L^{\rm em}$ for
any homogeneous BI-field history allowed by that action, provided there is no
connection-coupled spin current and no additional curvature or torsion
operators.  The minimal result is thus a direct specialization of the general
criterion rather than a property of one particular background solution.

The local tensor statement does not require $k_E>0$; that inequality is a
separate scalar-health condition.  The functions $B$, $Z$, and $U$ can still
affect GW observations indirectly by changing $x(z)$, the expansion history,
or source environments.  Equation~\eqref{eq:prop_theorem_distance} is a pure
\emph{propagation} relation and does not include possible changes to binary
waveform generation, screening, or detector/source couplings.  Likewise, the
criterion applies specifically to Eq.~\eqref{eq:mother_action}.  Extensions
with curvature-squared terms, derivative torsion, independent propagating
connection modes, Nieh--Yan terms, or direct spin/matter couplings to the
connection lie outside its scope.

The canonically normalized Einstein-frame field is defined by
\begin{equation}
 \frac{d\chi}{dx}=\Mpl\sqrt{k_E(x)}.
 \label{eq:canonical_field}
\end{equation}
Since matter couples to $A^{-1}g_E$, its dimensionless scalar coupling is
\begin{equation}
 \alpha(\chi)=-\frac{\Mpl}{2}\frac{d\ln A}{d\chi}
 =-\frac{A_{,x}}{2A\sqrt{k_E}}.
 \label{eq:matter_coupling}
\end{equation}
This expression, rather than a derivative with respect to a noncanonical
field, is the physical Einstein-frame coupling.  An even $A(x)$ automatically
gives $\alpha(0)=0$.

Equation~\eqref{eq:prop_theorem_action} gives $Q_T=\Mpl^2A$ and $c_T=1$.
The corresponding source-free Fourier-mode equation is
\begin{equation}
 \ddot h_\lambda+[3+\nu(t)]H\dot h_\lambda+
 \frac{k^2}{a^2}h_\lambda=0,
 \qquad \nu\equiv\frac{\dot A}{HA}=\frac{d\ln A}{d\ln a}.
 \label{eq:nonminimal_wave}
\end{equation}
In the geometric-optics limit the amplitude is proportional to
$(a\sqrt{A})^{-1}$, so that, with the same shorthand $A(z)=A[x(z)]$,
\begin{equation}
 \Xi(z)=\frac{d_L^{\rm gw}}{d_L^{\rm em}}
 =\sqrt{\frac{A_0}{A(z)}}
 =\exp\left[+\frac12\int_0^z\frac{\nu(z')}{1+z'}dz'\right].
 \label{eq:distance_map}
\end{equation}
Unlike Eq.~\eqref{eq:no_go}, this is a genuine nonminimal propagation effect.
For $A=\mathrm{const.}$, Eq.~\eqref{eq:distance_map} reduces identically to the
GR propagation law even if $B(x)$ and the torsion-induced scalar kinetic term
are strongly field dependent.

\paragraph{Proposition (Action-level tensor normalization criterion).}
For the curvature-linear first-order BI class considered here, after eliminating the algebraic Lorentz connection, the quadratic transverse-traceless tensor action can be written as
\[
S_T^{(2)}=\frac{1}{8}\int d^4x\,a^3 Q_T
\left(\dot h_{ij}^2-\frac{c_T^2}{a^2}(\partial_k h_{ij})^2\right),
\]
with
\[
Q_T=M_{\rm Pl}^{2}A(x),\qquad c_T^2=1 .
\]
Therefore, anomalous GW luminosity-distance evolution is controlled by the parity-even tensor normalization coefficient \(A(x)\), rather than by the Holst-to-Palatini ratio alone.

\subsection{Scalar principal part and stability conditions}
\label{sec:scalar_stability}

The Einstein-frame reduction also makes the scalar stability conditions
transparent.  Defining the canonical field by Eq.~\eqref{eq:canonical_field}
and the Einstein-frame potential by
\begin{equation}
 V_E(\chi)\equiv \frac{U[x(\chi)]}{A[x(\chi)]^2},
 \label{eq:VE_def}
\end{equation}
the gravitational--scalar sector becomes Einstein gravity plus a canonical
scalar.  Consequently, provided $A>0$ and $k_E>0$, the intrinsic Einstein-frame
scalar-field principal symbol has the standard sign and unit characteristic
speed,
\begin{equation}
\begin{aligned}
 S^{(2)}_{\delta\chi,\,\mathrm{principal}}={}&\frac12\int dt_E\,d^3x\,a_E^3
 \left[(\delta\dot\chi)^2-\frac{(\boldsymbol\nabla\delta\chi)^2}{a_E^2}\right],\nonumber\\
 &c_s^2=1.
\end{aligned}
 \label{eq:scalar_principal}
\end{equation}
Thus $A>0$ and $k_E>0$ exclude a tensor kinetic ghost and a scalar kinetic
ghost and give the canonical short-wavelength scalar-field gradient sign in
the gravity--scalar subsystem.  Matter degrees of freedom retain their own
principal characteristics, and the coupled cosmological system can contain
additional low-frequency or matter-sector instabilities.  The effective mass
$d^2V_E/d\chi^2$, mixing with matter perturbations, long-wavelength behavior,
and local screening constraints must therefore be checked separately.  The
sign of $k_J$ is not substituted for this invariant criterion.  We describe
the benchmark below as \emph{kinetically healthy at the two-derivative level},
not as a globally stable cosmology.

To motivate a minimal nonminimal completion, assume analyticity around the
parity-symmetric point $x=0$.  Parity requires $A$, $Z$, and $U$ to be even
and $B$ (hence $\gamma=B/A$) to be odd.  Fixing the present gravitational
normalization gives $A=1+\xi x^2+\mathcal O(x^4)$.  On the generic analytic
branch with $\gamma_{,x}(0)\neq0$, a nonsingular field rescaling can also set
$\gamma=x+\mathcal O(x^3)$.  (If the linear odd coefficient vanishes, the
leading Holst ratio starts at cubic or higher order and is a different EFT
branch.)  We retain the leading terms of the nondegenerate branch as a
concrete two-derivative benchmark.  The following ansatz is therefore not the
unique BI completion; it is the lowest-order analytic parity-invariant
representative of this branch that leaves the minimal Holst submanifold and
allows a running tensor normalization.

For the phenomenological catalog recast we use the standard two-parameter interpolation
\begin{equation}
 \Xi(z)=\Xi_0+\frac{1-\Xi_0}{(1+z)^n},
 \label{eq:Xi_ansatz}
\end{equation}
where $\Xi_0$ denotes the asymptotic high-redshift value and $n$ controls the transition rate.

For a concrete one-parameter curvature structure we take
\begin{equation}
 A(x)=1+\xi x^2,\qquad B(x)=A(x)x,\qquad Z(x)=1,
 \label{eq:benchmark_functions}
\end{equation}
so that $\gamma=x$ and
\begin{equation}
 k_E=\frac{1}{1+\xi x^2}+\frac{3}{2(1+x^2)},\qquad
 k_J=A\left[k_E-\frac{6\xi^2x^2}{A^2}\right].
 \label{eq:benchmark_kinetics}
\end{equation}
Equations~\eqref{eq:benchmark_functions}--\eqref{eq:benchmark_kinetics}
provide a definite microscopic matching for the GW damping.  A siren bound
constrains the evolution of $A=1+\xi x^2$, not $x$ or $\xi$ separately.

\begin{table*}[t]
\caption{Minimal and curvature-dressed BI sectors.  ``Algebraic torsion''
means that no independent connection mode propagates.  All entries are derived
in this work from Eqs.~\eqref{eq:minimal_first_order}--
\eqref{eq:benchmark_kinetics}; the minimal connection reduction follows
Refs.~\cite{Taveras2008,TorresGomez2009}.  No external numerical data enter
this table.}
\label{tab:models}
\centering
\begin{tabularx}{\textwidth}{@{}YYY@{}}
\toprule
Property & Minimal dynamical Holst & First-order completion \\
\midrule
Palatini coefficient & $A=1$ & $A(x)=1+\xi x^2$ \\
Holst coefficient & $B=\beta$ & $B=A x$ \\
Holst-to-Palatini role & scalar kinetic metric & scalar kinetic metric \\
Connection & algebraic torsion & algebraic torsion \\
Tensor normalization & $Q_T=\Mpl^2$ & $Q_T=\Mpl^2A$ \\
GW speed & $c_T=1$ & $c_T=1$ \\
Distance ratio & $\Xi=1$ & $\Xi=\sqrt{A_0/A}$ \\
Scalar kinetic condition & automatic & $A>0$, $k_E>0$ \\
Scalar principal speed & $c_s^2=1$ & $c_s^2=1$ \\
\bottomrule
\end{tabularx}
\end{table*}

\section{Background consistency}
\label{sec:background}

\subsection{Jordan-frame equations}

For pressureless matter and radiation, the flat-FLRW equations following
from Eq.~\eqref{eq:Jordan_action} are
\begin{align}
 3\Mpl^2AH^2={}&\rho_m+\rho_r+
 \frac{\Mpl^2}{2}k_J\dot x^2+U-3\Mpl^2H\dot A,
 \label{eq:Friedmann}\\
 -2\Mpl^2A\dot H={}&\rho_m+\frac43\rho_r+
 \Mpl^2k_J\dot x^2+\Mpl^2\ddot A-\Mpl^2H\dot A.
 \label{eq:Raychaudhuri}
\end{align}
They also show why an arbitrary curve $A(z)$ need not belong to a chosen
background subclass: after $H(z)$ is specified, Eq.~\eqref{eq:Raychaudhuri}
fixes the Jordan-frame combination $k_J\dot x^2$.  Its sign is useful for the
restricted diagnostic below, but is not by itself the invariant scalar
no-ghost test.

For a diagnostic, impose a flat $\Lambda$CDM expansion, neglect radiation at
$z=0$, set $A_0=1$, and use the ansatz \eqref{eq:Xi_ansatz} with
$A(z)=\Xi(z)^{-2}$.  Equation~\eqref{eq:Raychaudhuri} then gives
\begin{equation}
\begin{aligned}
 X_0\equiv\frac{k_J\dot x_0^2}{H_0^2}={}&-2nA_\Xi(1+\epsilon_0-n)-6n^2A_\Xi^2,\nonumber\\
 A_\Xi\equiv{}&1-\Xi_0,\qquad \epsilon_0\equiv\frac32\Omega_{m0}.
\end{aligned}
 \label{eq:X0_test}
\end{equation}
For $k_J>0$, $X_0\geq0$ is necessary.  Apart from the GR root, the boundary is
\begin{equation}
 \Xi_{\rm b}(n)=\frac23+\frac{1+\epsilon_0}{3n};
 \label{eq:feasible_boundary}
\end{equation}
the necessary band lies between $\Xi_0=1$ and $\Xi_{\rm b}$.  With
$\Omega_{m0}=0.3065$ and $n=1$, the often-used examples $\Xi_0=0.5$ and $1.9$
give $X_0=-1.96$ and $-4.03$, respectively.  They are useful kinematic shapes
but do not lie in the additionally restricted fixed-$\Lambda$CDM, $k_J>0$
subclass.  Equations~\eqref{eq:X0_test}--\eqref{eq:feasible_boundary} therefore
define only a present-day compatibility test for that subclass; they are not
a frame-invariant no-ghost criterion, a sufficient global viability test, or
a model-selection likelihood.

\subsection{A dynamically consistent existence benchmark}

We now demonstrate that the first-order completion is not empty.  For
Eq.~\eqref{eq:benchmark_functions}, choose $\xi=1$ and prescribe the monotonic
trajectory, for $N\leq0$,
\begin{equation}
\begin{aligned}
 x(N)={}&x_\star\tanh(q+cq^2),\qquad q=-\frac{N}{N_t},\nonumber\\
 x_\star={}&0.15,\qquad N_t=2.
\end{aligned}
 \label{eq:trajectory}
\end{equation}
We integrate Eq.~\eqref{eq:Raychaudhuri} for
$y(N)\equiv H^2/H_0^2$ with $y(0)=1$, $\Omega_{m0}=0.3065$, and
$\Omega_{r0}=9\times10^{-5}$.  In dimensionless form,
\begin{equation}
 y'=-\frac{3\Omega_{m0}e^{-3N}+4\Omega_{r0}e^{-4N}
 +y\left[k_Jx'^2+A''-A'\right]}{A+A'/2}.
 \label{eq:y_ode}
\end{equation}
The Friedmann equation then reconstructs
\begin{equation}
\begin{aligned}
 u(N)\equiv\frac{U}{\Mpl^2H_0^2}={}&3Ay-3\Omega_{m0}e^{-3N}
 -3\Omega_{r0}e^{-4N}\nonumber\\
 &-\frac12k_Jyx'^2+3yA'.
\end{aligned}
 \label{eq:potential_reconstruction}
\end{equation}
For Table~\ref{tab:benchmark} we use the Jordan-frame effective scalar
fraction
\begin{equation}
 \Omega_x\equiv
 \frac{\Mpl^2 k_J\dot x^2/2+U-3\Mpl^2H\dot A}
 {3\Mpl^2AH^2}
 =1-\frac{\rho_m+\rho_r}{3\Mpl^2AH^2}.
 \label{eq:Omega_x}
\end{equation}
This definition assigns the nonminimal term to the effective scalar sector;
other Jordan-frame energy-density splits are possible, so the convention is
stated explicitly.
The value $c=2.52741375$ enforces the parity boundary condition
$U_{,x}(0)=0$.  Because $x(N)$ is monotonic, Eq.~\eqref{eq:potential_reconstruction}
defines a single-valued potential on the sampled branch; it can be extended
as $U(-x)=U(x)$.  The scalar equation follows from the Bianchi identity for
$x'\ne0$, while the imposed boundary condition makes it regular at $x=0$.

\begin{figure*}[t]
\centering
\includegraphics[width=0.98\textwidth]{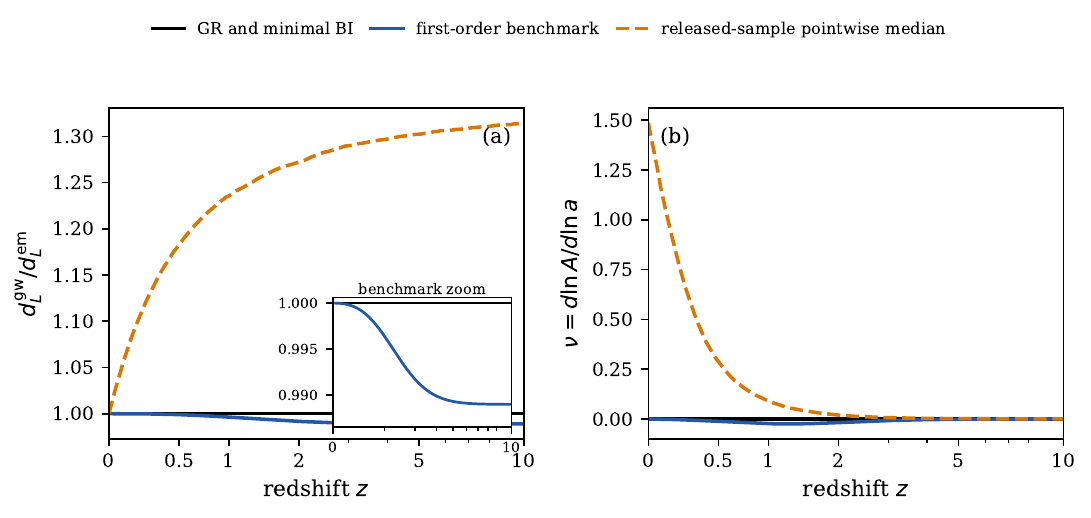}
\caption{Propagation observables.  (a) GR and the minimal BI theory have
$\Xi=1$.  The solid blue curve is the first-order benchmark of
Eq.~\eqref{eq:trajectory}; the inset resolves its percent-level change.  The
dashed curve is the pointwise weighted median of Eq.~\eqref{eq:Xi_ansatz}
computed from the released GWTC-4.0 samples, preserving the sampled
$\Xi_0$--$n$ covariance.  (b) The corresponding friction $\nu$.  The
kinematic posterior summary is not asserted to satisfy the background
equations.}
\label{fig:distance}
\end{figure*}

\begin{figure*}[t]
\centering
\includegraphics[width=0.96\textwidth]{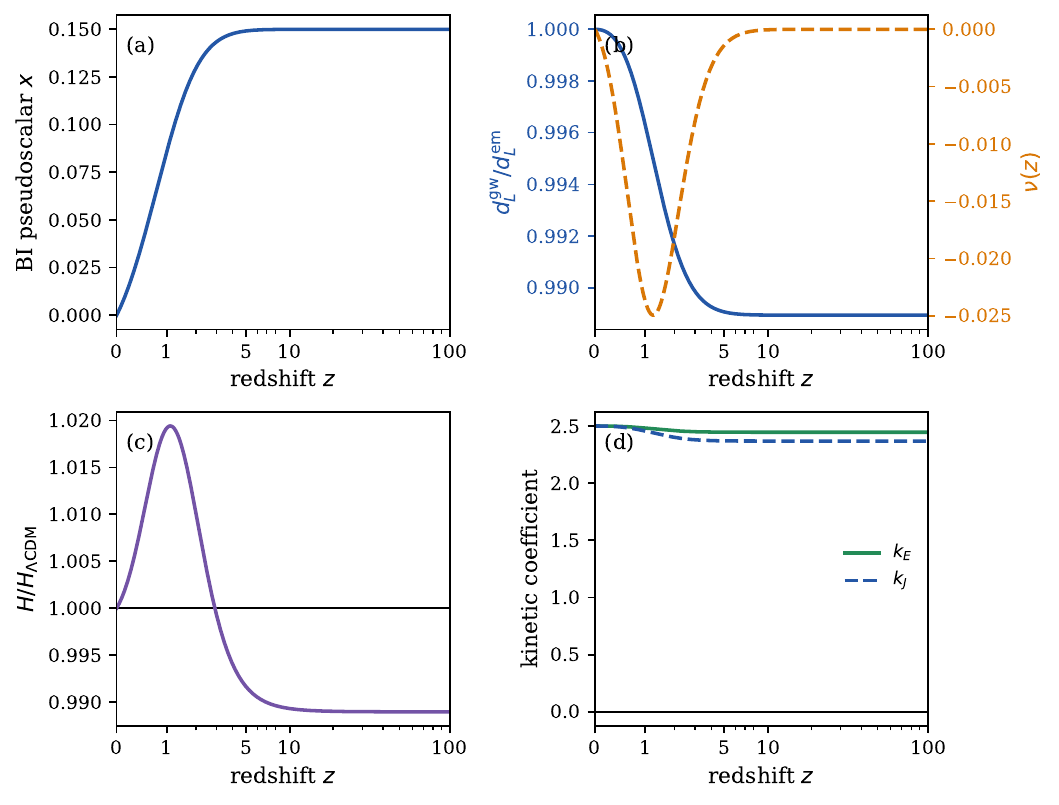}
\caption{Self-consistent background for Eqs.~\eqref{eq:benchmark_functions}
and \eqref{eq:trajectory}: (a) BI pseudoscalar, (b) distance ratio and friction,
(c) expansion relative to flat $\Lambda$CDM with the same present density
parameters, and (d) Einstein- and Jordan-frame kinetic functions.  The
integration extends to $z=1100$; $k_E$ and $k_J$ remain positive.}
\label{fig:benchmark}
\end{figure*}

The numerical checks are collected in Table~\ref{tab:benchmark}.  The model
has $\alpha_0=0$ and $\nu_0=0$ because $A_{,x}(0)=0$, even though the field
crosses $x=0$.  The reconstructed potential is positive on the sampled branch,
and $A>0$ and $k_E>0$ hold throughout $0\leq z\leq1100$; the optional stronger
property $k_J>0$ also happens to hold.  Together with
Eq.~\eqref{eq:scalar_principal}, this verifies the two-derivative kinetic and
intrinsic short-wavelength scalar-field gradient conditions checked here.  It
is an existence proof, not a fit to background, large-scale-structure, or
Solar-System data.  Its role is to
replace unsupported ``viable example'' curves with a solution that actually
satisfies Eqs.~\eqref{eq:Friedmann}--\eqref{eq:Raychaudhuri}.

\begin{table}[t]
\caption{Diagnostics for the self-consistent benchmark.  Source: numerical
integration of Eqs.~\eqref{eq:y_ode}--\eqref{eq:potential_reconstruction} in
this work; the tabulated values are recorded in
\texttt{output/background\_diagnostics.json}.}
\label{tab:benchmark}
\centering
\begin{tabular}{@{}lc@{}}
\toprule
Quantity & Value \\
\midrule
$\Xi(z\rightarrow\infty)$ & $0.988936$ \\
$\min A$ & $1.000000$ \\
$\min k_E$ & $2.44499$ \\
$\min k_J$ & $2.36797$ \\
$\min[U/(\Mpl^2H_0^2)]$ & $2.07320$ \\
$\max_{z\leq10}|H/H_{\Lambda{\rm CDM}}-1|$ & $1.94\%$ \\
$\Omega_x(z=10)$ & $2.39\times10^{-3}$ \\
$\Omega_x(z=100)$ & $2.97\times10^{-6}$ \\
$\alpha_0$, $\nu_0$ & $0$, $0$ \\
\bottomrule
\end{tabular}
\end{table}

\section{Standard-siren implications and catalog recast}
\label{sec:data}

\subsection{Catalog results and template-level endpoint map}

The catalog analysis below is deliberately a phenomenological application of
the propagation theorem, not a direct likelihood analysis of the specific
background in Sec.~\ref{sec:background}.  In particular, the explicit BI
benchmark predicts a full function $A(z)$, whereas the public catalog samples
were obtained with the two-parameter template~\eqref{eq:Xi_ansatz}.  We
therefore use the catalog products to identify which combinations of the
running tensor normalization are currently constrained and to expose prior and
population degeneracies; we do not reinterpret them as a posterior for
$\xi$ or for the BI field amplitude.

For $A_0=1$, the asymptotic quantities related to Eq.~\eqref{eq:Xi_ansatz}
are
\begin{equation}
 \frac{A_\infty}{A_0}=\Xi_0^{-2},\qquad
 \Delta_M\equiv\ln\frac{A_\infty}{A_0}=-2\ln\Xi_0.
 \label{eq:exact_recast}
\end{equation}
This monotonic transformation is exact within the phenomenological
parameterization~\eqref{eq:Xi_ansatz} and contains no effective-redshift or
Gaussian approximation.  It is not, however, a direct parameter constraint on
the explicit trajectory defined by Eqs.~\eqref{eq:benchmark_functions} and
\eqref{eq:trajectory}: the catalog likelihood depends on the full redshift
evolution of $A(z)$, so that model requires a dedicated inference.

The GWTC-3 population analysis of Ref.~\cite{Mancarella2022} reported the 90\%
intervals $\Xi_0\in[0.2,2.7]$ under a flat $\Xi_0$ prior and
$\Xi_0\in[0.1,1.9]$ under a flat $\ln\Xi_0$ prior.  Their exact endpoint maps
are shown in Table~\ref{tab:catalog_intervals}.  These are images of the
published intervals; they need not be highest-density intervals in the
transformed coordinate.

The final published GWTC-4.0 cosmology paper reports
\begin{equation}
 \Xi_0=1.4^{+1.0}_{-0.4}\ 
 \bigl(1.4^{+2.8}_{-0.6}\bigr),
 \label{eq:gwtc4_headline}
\end{equation}
where the first and parenthesized ranges contain 68.3\% and 90\% probability,
respectively \cite{GWTC4cosmology}.  Thus the final 90\% interval
$[0.8,4.2]$ maps to $A_\infty/A_0\in[0.0567,1.5625]$ and
$\Delta_M\in[-2.870,0.446]$.  GR remains inside all quoted intervals.

The GWTC-5.0 modified-propagation analysis, based on the FullPop-4.0 spectral-siren
population model with 235 GW candidates, reports \cite{GWTC5cosmology}
\begin{align}
 \Xi_0 &= 1.1^{+0.6}_{-0.3}\ \bigl(1.1^{+1.7}_{-0.5}\bigr),
 && \text{wide }H_0\text{ prior},\label{eq:gwtc5_wide}\\
 \Xi_0 &= 1.0^{+0.3}_{-0.2}\ \bigl(1.0^{+0.6}_{-0.3}\bigr),
 && \text{narrow }H_0\text{ prior},\label{eq:gwtc5_narrow}
\end{align}
where the first and parenthesized uncertainties are the symmetric 68.3\% and
90\% credible intervals, respectively.  Applying Eq.~\eqref{eq:exact_recast}
to the reported endpoints gives, for example, the 90\% maps
$\Delta_M\in[-2.059,1.022]$ (wide $H_0$ prior) and
$\Delta_M\in[-0.940,0.713]$ (narrow $H_0$ prior).  Both contain the GR point
$\Delta_M=0$.  These are exact coordinate maps of the published intervals, not
a reanalysis of the GWTC-5.0 likelihood or a posterior for the microscopic BI
parameters.

\begin{table*}[t]
\caption{Exact endpoint transformation of published catalog intervals.  The
GWTC-3 rows use Ref.~\cite{Mancarella2022}; the GWTC-4.0 rows use the final
published result of Ref.~\cite{GWTC4cosmology}; and the GWTC-5.0 rows use the
published FullPop-4.0 spectral-siren constraints of Ref.~\cite{GWTC5cosmology}.
The wide- and narrow-$H_0$ GWTC-5.0 analyses use different $H_0$ priors and
therefore are not interchangeable.  All transformed columns are calculated in
this work using Eq.~\eqref{eq:exact_recast}; they are endpoint maps rather than
new posterior intervals in the transformed coordinates.}
\label{tab:catalog_intervals}
\centering
\small
\setlength{\tabcolsep}{4pt}
\begin{tabular}{@{}llll@{}}
\toprule
Analysis & interval in $\Xi_0$ & interval in $A_\infty/A_0$ &
interval in $\Delta_M$ \\
\midrule
GWTC-3, flat $\Xi_0$ (90\%) & $[0.2,2.7]$ & $[0.137,25.0]$ & $[-1.987,3.219]$ \\
GWTC-3, flat $\ln\Xi_0$ (90\%) & $[0.1,1.9]$ & $[0.277,100]$ & $[-1.284,4.605]$ \\
GWTC-4.0 final (68.3\%) & $[1.0,2.4]$ & $[0.174,1.000]$ & $[-1.751,0]$ \\
GWTC-4.0 final (90\%) & $[0.8,4.2]$ & $[0.0567,1.5625]$ & $[-2.870,0.446]$ \\
GWTC-5.0, wide $H_0$ (68.3\%) & $[0.8,1.7]$ & $[0.346,1.5625]$ & $[-1.061,0.446]$ \\
GWTC-5.0, wide $H_0$ (90\%) & $[0.6,2.8]$ & $[0.1276,2.7778]$ & $[-2.059,1.022]$ \\
GWTC-5.0, narrow $H_0$ (68.3\%) & $[0.8,1.3]$ & $[0.5917,1.5625]$ & $[-0.525,0.446]$ \\
GWTC-5.0, narrow $H_0$ (90\%) & $[0.7,1.6]$ & $[0.3906,2.0408]$ & $[-0.940,0.713]$ \\
\bottomrule
\end{tabular}
\end{table*}

\subsection{GWTC-5.0 posterior recast and GWTC-4.0 release comparison}

The GWTC-5.0 rows in Table~\ref{tab:catalog_intervals} are transformations
of the published intervals only.  The sample-level covariance and prior-Jacobian
diagnostics below retain the two GWTC-4.0 hyperposterior products used in our
reproducibility archive, so that the provenance of every reported sample-level
number remains explicit and version-pinned.

We additionally analyze the two public GWTC-4.0 hyperposterior files
\path{icarogw_dark_Xi0CDM_multipop.json} and
\path{gwcosmo_dark_Xi0CDM_multipop.json} from Zenodo record 16919645
\cite{GWTC4data}.  For a transparent release-level comparison, we form an
illustrative mixture that assigns equal total weight to the two pipeline sample
sets.  This is not an official LVK pipeline combination.  For every sample we
calculate Eq.~\eqref{eq:exact_recast}; no random surrogate samples are
generated.  The pipeline-specific $\Xi_0$ medians are 1.214 (\texttt{icarogw})
and 1.497 (\texttt{gwcosmo}), which motivates keeping their identities visible
rather than treating the mixture as a new catalog measurement.

The release used uniform priors $\Xi_0\in[0.435,10]$ and $n\in[0.1,10]$.
A posterior quoted under a prior uniform in $\Delta_M$ is a different
inference.  On the same likelihood samples it is obtained by importance weights
\begin{equation}
 w_{\rm flat\ \Delta_M}\propto w_{\rm release}
 \left|\frac{d\Delta_M}{d\Xi_0}\right|
 \propto\frac{w_{\rm release}}{\Xi_0}.
 \label{eq:prior_reweight}
\end{equation}
The reweighting is restricted to the inherited support
$\Delta_M\in[-2\ln 10,-2\ln 0.435]$ and performs no extrapolation beyond the
released samples.
The resulting effective sample size is 8611 out of 11314, so the numerical
reweighting is not dominated by a handful of samples.  Table~\ref{tab:sample_results}
reports equal-tailed summaries for each pipeline and for the mixture.  The
pipeline-specific 68.3\% intervals, $[0.839,1.993]$ for \texttt{icarogw} and
$[0.990,3.288]$ for \texttt{gwcosmo}, show that the released analyses differ
most strongly in the high-$\Xi_0$ tail.  The additional shift between the two
mixture rows is a property of the chosen prior measure, not new observational
information.

The released pipeline files and the final journal result are versioned data
products.  Their equal-weight sample mixture has median $\Xi_0=1.336$ and
does not exactly reproduce Eq.~\eqref{eq:gwtc4_headline}.  We therefore use
the final paper for the official headline constraint and use the named public
files only for posterior-shape, covariance, Jacobian, and feasibility
diagnostics.  This separation prevents a release-level calculation from
being presented as an exact reconstruction of the final collaboration
combination.

\begin{table*}[t]
\caption{Sample-level recast of the public GWTC-5.0 FullPop spectral-siren posterior samples. Entries are the median, central 68.3\% interval, and central 90\% interval after transforming the released $\Xi_0$ samples.}
\label{tab:gwtc5_sample_results}
\centering
\small
\begin{tabular}{@{}llll@{}}
\toprule
Sample & $\Xi_0$ & $\Delta_M$ & $A_\infty/A_0$ \\
\midrule
Wide $H_0$ prior & $1.097^{+0.577}_{-0.322}$ & $-0.184^{+0.510}_{-1.030}$ & $0.832^{+0.833}_{-0.475}$ \\
90\% interval & $[0.611,2.805]$ & $[-2.063,0.986]$ & $[0.127,2.680]$ \\
Narrow $H_0$ prior & $1.015^{+0.286}_{-0.180}$ & $-0.029^{+0.361}_{-0.526}$ & $0.971^{+0.464}_{-0.380}$ \\
90\% interval & $[0.729,1.662]$ & $[-1.017,0.633]$ & $[0.362,1.884]$ \\
\bottomrule
\end{tabular}
\end{table*}

\begin{table*}[t]
\caption{Sample-level recast of the two named GWTC-4.0 public posterior files,
their illustrative equal-pipeline mixture, and the prior-reweighted mixture.
Source: the \texttt{icarogw} and \texttt{gwcosmo} hyperposteriors in
Ref.~\cite{GWTC4data}.  Entries are the median, central 68.3\% interval, and
central 90\% interval computed with the accompanying script.  These
release-sample summaries are distinct from the final published headline in
Eq.~\eqref{eq:gwtc4_headline}.}
\label{tab:sample_results}
\centering
\small
\setlength{\tabcolsep}{4pt}
\begin{tabular}{@{}llll@{}}
\toprule
Sample and prior & $\Xi_0$ & $\Delta_M$ & $A_\infty/A_0$ \\
\midrule
\texttt{icarogw}, flat $\Xi_0$ &
$\shortstack{$1.214^{+0.779}_{-0.375}$\\[1pt]\scriptsize$[0.670,3.596]$}$ &
$\shortstack{$-0.387^{+0.739}_{-0.992}$\\[1pt]\scriptsize$[-2.559,0.801]$}$ &
$\shortstack{$0.679^{+0.743}_{-0.427}$\\[1pt]\scriptsize$[0.0774,2.228]$}$ \\
\texttt{gwcosmo}, flat $\Xi_0$ &
$\shortstack{$1.497^{+1.791}_{-0.507}$\\[1pt]\scriptsize$[0.785,5.552]$}$ &
$\shortstack{$-0.807^{+0.826}_{-1.574}$\\[1pt]\scriptsize$[-3.428,0.484]$}$ &
$\shortstack{$0.446^{+0.574}_{-0.354}$\\[1pt]\scriptsize$[0.0324,1.623]$}$ \\
Equal-pipeline mixture, flat $\Xi_0$ &
$\shortstack{$1.336^{+1.283}_{-0.435}$\\[1pt]\scriptsize$[0.716,4.727]$}$ &
$\shortstack{$-0.580^{+0.788}_{-1.346}$\\[1pt]\scriptsize$[-3.107,0.668]$}$ &
$\shortstack{$0.560^{+0.672}_{-0.414}$\\[1pt]\scriptsize$[0.0447,1.950]$}$ \\
Equal-pipeline mixture, flat $\Delta_M$ &
$\shortstack{$1.112^{+0.586}_{-0.334}$\\[1pt]\scriptsize$[0.621,2.675]$}$ &
$\shortstack{$-0.213^{+0.713}_{-0.847}$\\[1pt]\scriptsize$[-1.968,0.953]$}$ &
$\shortstack{$0.808^{+0.841}_{-0.462}$\\[1pt]\scriptsize$[0.140,2.593]$}$ \\
\bottomrule
\end{tabular}
\end{table*}

\begin{figure*}[t]
\centering
\includegraphics[width=0.98\textwidth]{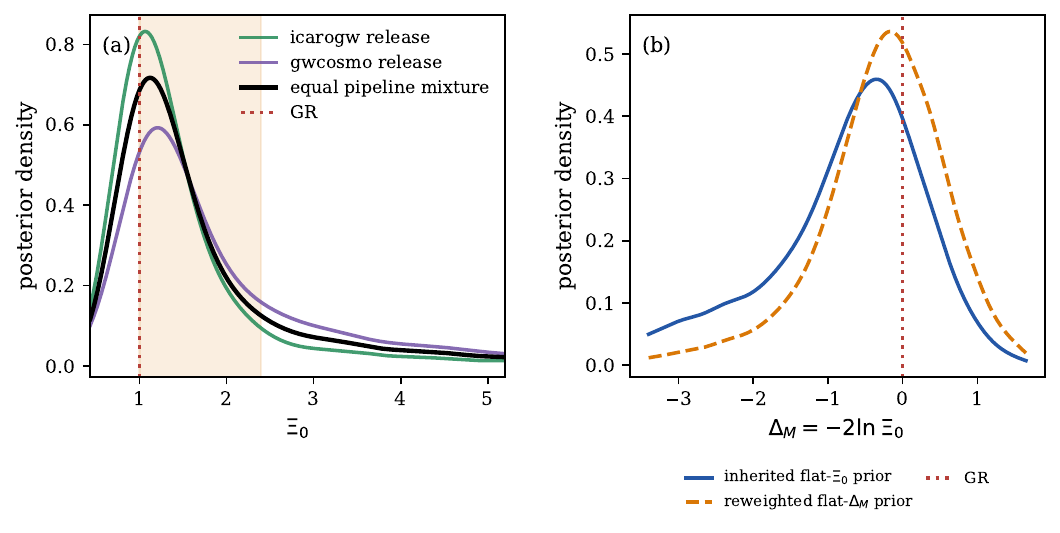}
\caption{GWTC-4.0 release-sample analysis.  (a) Pipeline-specific
$\Xi_0$ densities and their equal-weight mixture; the light orange region is
the final published 68.3\% interval $[1.0,2.4]$ and the vertical dotted line is
GR.  (b) Sample-by-sample transformation to $\Delta_M$ under the inherited
flat-$\Xi_0$ prior and after the Jacobian reweighting of
Eq.~\eqref{eq:prior_reweight}.}
\label{fig:posterior}
\end{figure*}

\begin{figure*}[t]
\centering
\includegraphics[width=0.82\textwidth]{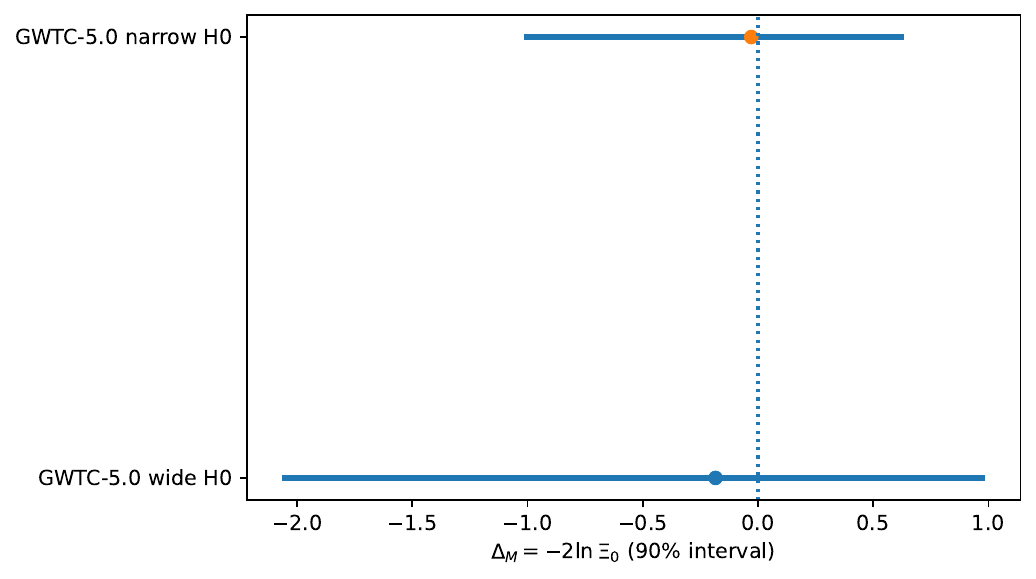}
\caption{Ninety-percent intervals expressed in the common coordinate
$\Delta_M$.  The GWTC-5.0 rows show the published FullPop-4.0 spectral-siren
results for the wide- and narrow-$H_0$ priors.  Markers are drawn only where a
central reported value or a sample median is available; no interval midpoint is
interpreted as a posterior estimate.}
\label{fig:catalogs}
\end{figure*}

\section{Restricted background-consistency diagnostic and parameter
degeneracies}
\label{sec:posterior_feasibility}

As a deliberately restrictive diagnostic, for each released sample
$(\Xi_0,n)$ we evaluate Eq.~\eqref{eq:X0_test} using
$\Omega_{m0}=0.3065$.  Only 19.1\% of the illustrative equal-pipeline posterior mass passes the
present-day condition $X_0\geq0$ after imposing the additional fixed-$\Lambda$CDM,
$k_J>0$ subclass assumption.  Figure~\ref{fig:feasibility} displays both the
analytic boundary and the conditional density.  This retained fraction is
neither a Bayes factor nor a posterior probability for scalar stability: the
frame-invariant scalar kinetic condition is $k_E>0$, and the original catalog
analysis imposed neither the BI background equations nor $k_J>0$.  A
model-consistent inference would require a new joint likelihood including the
background and perturbations.  Points outside the band therefore fail only
this restricted compatibility test; they are not excluded from the full
action class or from scalar--tensor theories in general.

\begin{figure*}[t]
\centering
\includegraphics[width=0.98\textwidth]{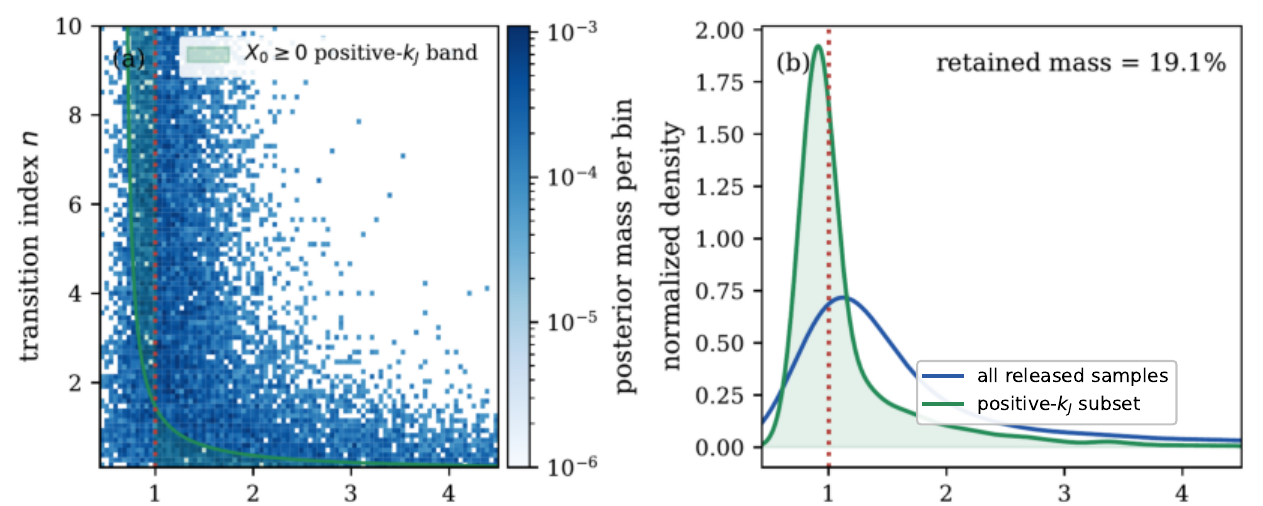}
\caption{Fixed-$\Lambda$CDM, positive-$k_J$ subclass diagnostic for the
released samples.
(a) Joint $(\Xi_0,n)$ posterior mass and the analytic $X_0\geq0$ necessary
band of Eq.~\eqref{eq:feasible_boundary}.  (b) Marginal $\Xi_0$ density before
and after conditioning on this test.  The 19.1\% retained mass is descriptive
and is not a model probability.}
\label{fig:feasibility}
\end{figure*}

The equal-mixture weighted correlations are
\begin{equation}
\begin{aligned}
 \rho(\Xi_0,H_0)&=0.308,\qquad \rho(\Xi_0,n)=-0.433,\nonumber\\
 \rho(\Xi_0,\gamma_{\rm pop})&=0.746.
\end{aligned}
 \label{eq:correlations}
\end{equation}
where $\gamma_{\rm pop}$ is the population-model parameter named
\texttt{gamma} in the release.  These correlations quantify why the
constraint cannot be interpreted as a one-dimensional measurement of a
microscopic BI coupling.  The luminosity-distance modification changes the
inferred source redshifts and therefore covaries with both cosmology and the
mass/redshift population model.

\section{Discussion}
\label{sec:discussion}

\subsection{Theoretical scope}

The main conclusions apply at different levels of generality.  The minimal-sector propagation result follows directly from the action and is exact within the
two-derivative dynamical-Holst theory without additional curvature couplings.
It does not depend on the BI potential, the background solution, or a
small-field expansion.  The first-order completion is more specific: it
selects the parity-even functions in Eq.~\eqref{eq:benchmark_functions} and
neglects fermions, Nieh--Yan couplings, curvature-squared operators, and direct
matter couplings.  Those extensions may produce additional observables, but
they must be matched at the action level rather than absorbed into an
effective friction by assertion.

The existence benchmark establishes regularity, $A>0$ and $k_E>0$, and the
canonical intrinsic scalar-field principal sign on one cosmological
trajectory; it additionally has $k_J>0$.  It does not establish a positive
effective mass, stability of the complete matter--scalar perturbation system,
nonlinear screening, radiative stability, or agreement with all cosmological
data.  A
full inference should solve the background and perturbation equations for
each parameter point and combine sirens with CMB, baryon-acoustic-oscillation,
supernova, growth, and local-gravity information.  The property
$A_{,x}(0)=0$ is useful because it gives $\alpha_0=\nu_0=0$, but local bounds
also depend on the scalar mass, environmental profile, and nonlinear dynamics.

\subsection{Comparison with related work}

Previous studies of a dynamical BI field developed its torsion-induced
scalar interactions, cosmological dynamics, couplings to fermions and
topological densities, and in some nonminimal torsionful models the morphology
of GW polarizations
\cite{Taveras2008,TorresGomez2009,Calcagni2009,MercuriTaveras2009,
Mercuri2009,Shaposhnikov2020,Karananas2021,Langvik2021,Rigouzzo2022,
Bombacigno2018,Bombacigno2019,Karananas2024}.  These results provide important
context but address different action classes or observables.  The specific
question isolated here is which coefficient of the curvature-linear BI action
controls the cosmological tensor amplitude after the algebraic connection has
been integrated out, and how a standard-siren constraint should be matched to
that coefficient.  Conversely, the standard modified-propagation literature
has developed the luminosity-distance parameterization and constrained it with
dark and bright sirens
\cite{Belgacem2018a,Belgacem2018b,Belgacem2019,Lagos2019,Finke2021,
Mastrogiovanni2021,Mukherjee2021,Mancarella2022,Leyde2022,ChenGray2024},
but generally treats the running gravitational coefficient
phenomenologically rather than deriving a BI first-order completion and
testing whether its background is dynamically realizable.  The present work combines these ingredients in one calculation: an exact
minimal-theory propagation result, an action-level propagation proposition for the
curvature-linear algebraic-connection class, explicit matching to a nonminimal
completion, a self-consistent background test, and a sample-level catalog
recast.  The resulting advance is therefore not another bound on the same
damping curve, but an action-level criterion identifying which BI action sectors
can consistently be constrained by gravitational-wave observations.

Phenomenological constructions that assign an additional GW damping while
retaining a minimal gravitational coefficient require particular care.
The tensor-sector reduction shows that such a propagation term cannot arise from
minimal algebraic BI torsion alone.  A
catalog-level damping constraint is physically interpretable in the BI sector
only after it is matched to an explicit nonminimal operator such as
Eq.~\eqref{eq:mother_action}.  This action-level distinction removes the
apparent conflict between a torsion-modified background and an unmodified
tensor normalization.  It also supersedes the minimal-sector
phenomenological-friction interpretation used in an earlier version of the
present work and refines the interpretation of the related BI--quintessence
extension \cite{GaoQuintessence2026}: within the minimal dynamical-Holst action
such a propagation term is not generated by the BI field alone, whereas a
running-friction description can consistently arise after matching to a
nonminimal completion with $\dot A\neq0$.

\subsection{Observational interpretation and limitations}

The statistical recast also has a deliberately limited meaning.  A monotonic
endpoint map is adequate for translating a reported interval, whereas a new
posterior coordinate requires its samples and prior measure.  Our
sample-by-sample calculation exposes both the posterior shape and the
Jacobian effect.  It does not reanalyze event strain data, reconstruct the
collaboration selection function, or supersede the final GWTC-4.0 result.
The released-sample $k_J>0$ cut is only the restricted compatibility
diagnostic defined above; it is neither a complete invariant stability condition nor
an observational exclusion quoted by LVK.

The observational information is nevertheless more revealing than a
single-parameter interval.  GWTC-5.0 tightens the published $\Xi_0$ constraint
relative to GWTC-4.0, with the LVK analysis reporting a 35.7\% improvement for
the wide-$H_0$ prior and a 50\% improvement for the narrow-$H_0$ prior, while
remaining consistent with GR \cite{GWTC5cosmology}.  At the same time,
Table~\ref{tab:sample_results} shows that the two public GWTC-4.0 pipelines
have visibly different high-$\Xi_0$ tails, while Fig.~\ref{fig:posterior}
demonstrates that the inferred $\Delta_M$ distribution changes under a
different prior measure.  The correlations with the transition index and
compact-binary population parameters in Eq.~\eqref{eq:correlations} explain
why a tighter phenomenological propagation interval is not by itself a direct
measurement of a microscopic BI coupling.  Near-term progress therefore
requires more well-localized and higher-redshift sirens, better calibration of
source-population features and selection effects, and a model-specific joint
likelihood rather than merely a tighter one-parameter fit
\cite{Finke2021,Mastrogiovanni2021,Mukherjee2021,ChenGray2024,GWTC5cosmology}.

\subsection{Relation to phenomenological gravitational-wave propagation tests}
\label{sec:relation_pheno_sirens}

Phenomenological standard-siren analyses usually introduce a propagation
function or an effective tensor-normalization evolution and constrain possible
departures from general relativity directly from gravitational-wave catalogs.
Such analyses determine the observational parameter space, but the microscopic
origin of the effective propagation parameters is not specified.  In contrast,
the present work starts from a first-order Einstein--Cartan--Holst action and
identifies which coefficients of the fundamental gravitational action survive
as observable tensor-propagation operators after the Lorentz connection is
eliminated.

The distinction is essential for interpreting standard-siren constraints.
Within the minimal algebraic Holst sector, the Barbero--Immirzi ratio affects
the scalar sector rather than introducing an independent tensor propagation
operator.  Observable deviations of the gravitational-wave luminosity distance
therefore probe an extended parity-even curvature sector rather than directly
measuring the minimal Barbero--Immirzi parameter.

This action-level interpretation differs from approaches that first assume a
phenomenological propagation deviation and then constrain its parameters.  It
provides a microscopic criterion for identifying which first-order gravity
extensions can produce observable gravitational-wave signatures.

\paragraph{Relation to scalar--tensor propagation mechanisms.}
It is important to distinguish the present action-level mechanism from the
usual scalar--tensor interpretation of modified gravitational-wave damping.
In many scalar--tensor theories the running effective Planck mass directly
changes the tensor kinetic normalization.  In the first-order BI framework
studied here, however, the minimal Holst ratio and the tensor normalization
coefficient are independent structures: the former modifies the scalar sector,
whereas observable tensor-amplitude evolution requires a nonminimal
parity-even sector with a varying coefficient \(A(x)\).  Thus the existence of
a dynamical BI scalar alone is not sufficient to imply anomalous GW friction.

\paragraph{Summary of the action-level interpretation.}
The results obtained in this work establish three points. First, in the minimal
dynamical Holst sector the BI field does not independently modify the tensor
propagation amplitude, because the tensor kinetic normalization remains fixed.
Second, in the broader curvature-linear first-order BI class, the coefficient
multiplying the parity-even curvature invariant determines the tensor kinetic
normalization,
\[
Q_T=M_{\rm Pl}^{2}A(x),
\]
and therefore controls the gravitational-wave luminosity-distance evolution.
Third, standard-siren observations constrain this evolving tensor
normalization sector rather than the minimal BI torsion dynamics itself.

\section{Conclusion and Outlook}
\label{sec:conclusion}

We have established an action-level criterion for GW propagation in the
curvature-linear, algebraic-connection bosonic class of dynamical BI theories
considered in this work.  The key result is that the
Hilbert--Palatini coefficient and the Holst-to-Palatini ratio play different
roles after the algebraic connection is eliminated.  The former normalizes the
tensor kinetic operator, whereas the latter fixes a calculable contribution to
the scalar kinetic metric.  Hence a dynamical Holst coefficient at fixed
Hilbert--Palatini normalization does not generate an additional GW friction contribution.  The
minimal dynamical-Holst theory is an exact corollary: it reduces to GR plus a
minimally coupled pseudoscalar, has luminal tensor speed. Consequently, gravitational and electromagnetic
luminosity distances remain identical in this minimal sector.

A GW-distance modification requires leaving this minimal submanifold.  We did
so with the leading parity-invariant first-order completion, for which the BI
pseudoscalar also controls a nonminimal parity-even curvature coefficient.
Connection elimination fixes the torsional contribution to the effective
couplings, while the tensor distance is controlled by the evolution of the
effective tensor normalization.  We supplied a reconstructed background satisfying both Friedmann equations
and the invariant two-derivative kinetic conditions $A>0$ and $k_E>0$.  The
Einstein-frame principal scalar action establishes the canonical intrinsic
short-wavelength scalar-field gradient sign, while the optional property
$k_J>0$ is kept conceptually separate.  A distinct analytic test then shows
that commonly plotted kinematic curves need not belong to the restricted
fixed-$\Lambda$CDM, $k_J>0$ subclass.

Finally, the published GWTC-3, GWTC-4.0, and GWTC-5.0 intervals were
mapped exactly to $A_\infty/A_0$ and $\Delta_M$.  The newer GWTC-5.0
spectral-siren bounds are tighter and remain consistent with the GR point.  The
public GWTC-4.0 pipeline samples were transformed individually, revealing
substantial prior dependence and strong population covariance.  Under the
additionally restricted fixed-$\Lambda$CDM, $k_J>0$ diagnostic,
19.1\% of the released posterior mass passes the present-day compatibility
condition; this is neither a model/stability probability nor an observational
exclusion.  Current sirens therefore provide a consistency test for the mapping between
observational tensor-normalization evolution and microscopic BI action sectors,
rather than a direct measurement of an isolated BI parameter. This distinction
separates the present action-based approach from purely phenomenological
modified-friction fits, while the minimal dynamical-Holst theory lies on the
exact GR-propagation surface.

The next observational step is a model-specific hierarchical analysis in
which the phenomenological propagation description is replaced by the numerically
generated tensor-normalization evolution of the BI completion at every sampled parameter point.  The
same background must then enter the source-distance relation, selection
function, source-frame masses, and joint inference of cosmological and
population parameters.  Larger samples of bright sirens, deeper and more
complete galaxy catalogs, and the extended redshift reach of next-generation
detectors should help break the propagation--population degeneracies exposed
here \cite{Finke2021,Mukherjee2021,ChenGray2024,Bertheas2026,DeLeo2026,
Nanadoumgar2026}.

On the theory side, the scalar and metric perturbations of the completion must
be evolved and confronted jointly with cosmic microwave background,
baryon-acoustic-oscillation, supernova, structure-growth, and local-gravity
constraints.  Extending the first-order action to include fermions,
Nieh--Yan terms, higher-curvature operators, and waveform-generation effects
will determine which observables can separate microscopic BI physics from a
generic running tensor normalization.  These steps can turn the present
action-consistency criterion into a direct inference of well-defined BI
parameters.

\section*{Data and Code Availability}
The two public GWTC-4.0 posterior products analyzed here are available from
the official LVK data release in Ref.~\cite{GWTC4data}.  They are not
redistributed with this manuscript.  A separate reproducibility archive
accompanying the submission contains the analysis scripts, a checksum-verifying
downloader for the exact public JSON files, the deterministic background grid
and diagnostics, and regression targets for the rounded catalog summaries.
Running the full analysis after downloading the public files regenerates the
sample transformations and all five figures.  The archive records all file
names and checksums needed to pin the analysis to the cited data-product
version.  The public GWTC-5.0 FullPop spectral-siren posterior samples used for the
recast are obtained from the official data release cited in Ref.~\cite{GWTC5cosmology};
the raw LVK posterior files are not redistributed.
as a new likelihood analysis in this work.

\section*{Statements and Declarations}
\textbf{Funding.} This research was supported by the National Natural Science Foundation of China (NSFC) project No.~12288102, the National Key Research and Development Program of China (2022YFC2205202), the NSFC (12573052, 12573103, 12373114, 12003009), and the Tianshan Talents Program (2023TSYCTD0013).\par
\textbf{Competing interests.} The authors declare no competing interests.

\appendix

\section{Formal propagation theorem and proof}
\label{app:propagation_theorem}

For completeness we collect here the formal statement used in
Sec.~\ref{sec:completion}.

\begin{proposition}[Tensor propagation in the curvature-linear algebraic-connection class]
\label{prop:tensor_propagation}
Consider the action~\eqref{eq:mother_action} with an independent antisymmetric
Lorentz connection $\omega^{IJ}=-\omega^{JI}$ and bosonic matter that does not
couple directly to $\omega^{IJ}$.  Let the background be spatially flat FLRW
with a homogeneous $x(t)$, assume $A(x)>0$, and assume that the connection
equation is algebraic and is eliminated using its regular solution
(equivalently, Eq.~\eqref{eq:general_contorsion} after the Einstein-frame
tetrad redefinition).  Then, for linear transverse-traceless metric
perturbations (and in the absence of a tensor anisotropic-stress source), the
exact two-derivative quadratic tensor action is Eq.~\eqref{eq:prop_theorem_action},
with $Q_T$, $c_T$, and $\nu$ given by Eq.~\eqref{eq:prop_theorem_local}.  In
the geometric-optics regime, comparison with the electromagnetic luminosity
distance in the same Jordan-frame background yields
Eq.~\eqref{eq:prop_theorem_distance}.  Thus $B(x)$ can affect tensor
observables indirectly through the background solution for $x$, but it does
not supply an independent tensor kinetic, gradient, or friction operator
within Eq.~\eqref{eq:mother_action}.
\end{proposition}

\begin{proof}
Introduce the Einstein-frame tetrad $e_E^I=\sqrt{A}\,e^I$, for which the
curvature part of Eq.~\eqref{eq:mother_action} is the Einstein--Cartan--Holst
operator with ratio $\gamma=B/A$.  Because the connection is algebraic, its
elimination produces Eq.~\eqref{eq:Einstein_action}: the $\gamma$ dependence
is transferred to the scalar kinetic coefficient $k_E$, while the metric
curvature term is the Einstein--Hilbert term.  Transforming back to the matter
frame gives Eq.~\eqref{eq:Jordan_action}.  At the stated derivative order its
only metric-curvature operator is $A(x)R$; all algebraic-torsion remnants lie
in the scalar sector.  Expanding $A(x)R$ around homogeneous FLRW gives
Eq.~\eqref{eq:prop_theorem_action}, from which Eq.~\eqref{eq:prop_theorem_local}
follows.  In geometric optics the canonically normalized tensor amplitude
redshifts as $(a\sqrt{A})^{-1}$, yielding Eq.~\eqref{eq:prop_theorem_distance}.
\end{proof}

\begin{corollary}[Minimal dynamical Holst propagation result]
\label{cor:minimal_no_go}
In the minimal bosonic two-derivative dynamical-Holst theory, $A=1$
identically.  For any homogeneous BI-field history allowed by that action, and
with no connection-coupled spin current or additional curvature or torsion
operators, $Q_T=\Mpl^2$, $c_T=1$, $\nu=0$, and the propagation-defined
luminosity distances satisfy $d_L^{\rm gw}=d_L^{\rm em}$.
\end{corollary}

\section{Connection elimination and kinetic matching}
\label{app:connection}

Factor $A(x)$ out of the curvature part of Eq.~\eqref{eq:mother_action} and
perform $g^E_{\mu\nu}=A g_{\mu\nu}$.  Up to a boundary term, the connection
dependent action is quadratic and linear in contorsion.  Varying it gives
\begin{equation}
\begin{aligned}
 C_{\mu IJ}=-\frac{1}{2[1+\gamma(x)^2]}\bigl(&\epsilon_{IJ}{}^{KL}e^E_{\mu K}\partial_L\gamma\nonumber\\
 &-2\gamma e^E_{\mu[I}\partial_{J]}\gamma\bigr).
\end{aligned}
 \label{eq:general_contorsion}
\end{equation}
Substituting Eq.~\eqref{eq:general_contorsion} produces
$-(3\Mpl^2/4)\gamma_{,x}^2(\partial x)^2/(1+\gamma^2)$.
The explicit scalar kinetic term transforms to
$-(\Mpl^2/2)(Z/A)(\partial x)^2$, yielding Eq.~\eqref{eq:kE_general}.
The inverse Weyl transformation obeys
\begin{equation}
 \sqrt{-g_E}R_E=\sqrt{-g}\left[A R+
 \frac{3}{2A}(\partial A)^2\right]+\hbox{boundary},
\end{equation}
which gives Eq.~\eqref{eq:kJ}.  This derivation also shows why the metric
operator $A(x)R$ and the torsion-induced term cannot be varied independently
once a mother action has been specified, and it supplies the connection-level
input used in Proposition~\ref{prop:tensor_propagation}.

\section{Background reconstruction details}
\label{app:background}

Using $d/dt=H d/dN$, one has
$\dot H/H^2=y'/(2y)$, $\dot A=HA'$, and
$\ddot A=H^2[A''+(y'/2y)A']$.  Substitution into
Eq.~\eqref{eq:Raychaudhuri} gives Eq.~\eqref{eq:y_ode}; substitution into
Eq.~\eqref{eq:Friedmann} gives Eq.~\eqref{eq:potential_reconstruction}.

At $N=0$, $x=0$ and $A_{,x}=k_{J,x}=0$.  A parity-even reconstructed
potential requires $U_{,x}(0)=0$.  The scalar equation then reduces to
\begin{equation}
 x''_0+\left(3+\frac{y'_0}{2}\right)x'_0=0.
 \label{eq:parity_boundary}
\end{equation}
For Eq.~\eqref{eq:trajectory}, this condition fixes
$c=N_t(3+y'_0/2)/2=2.52741375$.  Numerically the residual of
Eq.~\eqref{eq:parity_boundary} is below machine precision.  The supplied CSV
contains $x$, $A$, $H/H_0$, $U/(\Mpl^2H_0^2)$, $k_E$, and $k_J$ at 5000
points from $z=1100$ to the present.

\section{Posterior transformations and interval conventions}
\label{app:posterior}

For normalized sample weights $w_i$, weighted quantiles are calculated from
the ordered cumulative weights.  The equal-pipeline mixture assigns total
weight $1/2$ to each public file, not equal weight to each sample across files.
All derived coordinates are evaluated at the same sample index:
\begin{equation}
 \Delta_{M,i}=-2\ln\Xi_{0,i},\qquad
 (A_\infty/A_0)_i=\Xi_{0,i}^{-2}.
\end{equation}
The reweighted effective sample size is
$N_{\rm eff}=(\sum_i\tilde w_i^2)^{-1}$ after normalization.  A transformed
highest-density interval is not generally the highest-density interval in the
new coordinate because the probability density acquires a Jacobian.  This is
why Table~\ref{tab:catalog_intervals} is labelled an endpoint map, whereas
Table~\ref{tab:sample_results} is a sample-level posterior summary.

\section{Reproducibility}
\label{app:reproducibility}

The reproducibility archive contains four scripts with separated roles.
\path{code/analysis_core.py} implements the equations and weighted
statistics; \path{code/theory_checks.py} independently checks the background
numbers, the distance-map integral identity, the analytic $X_0=0$ boundary,
and the prior Jacobian without requiring catalog downloads;
\path{code/download_public_data.py} retrieves the two versioned public
GWTC-4.0 files and verifies their MD5 hashes; and
\path{code/analyze_and_plot.py} performs the complete sample-level recast,
integrates Eq.~\eqref{eq:y_ode}, reconstructs
Eq.~\eqref{eq:potential_reconstruction}, writes the derived tables, and
regenerates all five figures.  The public source is Zenodo record 16919645,
cited in Ref.~\cite{GWTC4data}.  The release-listed MD5 checksums are
\texttt{6b51973c94483338\allowbreak d236b933d7857e81} for the icarogw file and
\texttt{db655ba66532b19c\allowbreak b4a5333b7c607387} for the gwcosmo file.
The raw JSON files are intentionally not duplicated in the archive.  Once they
are downloaded, the option \texttt{--verify-\allowbreak manuscript} compares regenerated
posterior summaries with the rounded values printed in the paper.

\begingroup\sloppy

\endgroup

\end{document}